\documentclass[12pt]{amsart}
\usepackage{amssymb,amsmath}
\usepackage{psfrag,graphicx}
\usepackage{pdfsync}
\def\eps{\varepsilon}

\def\d{{\rm d}}
\def\ddt{\frac{\d}{\d t}}

\def\R {\mathbb{R}}
\def\N {\mathbb{N}}
\def\BB {\mathfrak{B}}
\def\H {{\mathcal H}}

\def\C {{\mathcal C}}
\def\D {{\mathcal D}}

\def\E {{\mathcal E}}
\def\L {{\mathcal L}}
\def\A {{\mathcal A}}

\def\Q {{\mathcal Q}}

\def \l {\langle}
\def \r {\rangle}
\def \pt {\partial_t}
\def \ptt {\partial_{tt}}
\newtheorem{proposition}{Proposition}[section]
\newtheorem{theorem}[proposition]{Theorem}

\newtheorem{lemma}[proposition]{Lemma}
\theoremstyle{definition}

\newtheorem{remark}[proposition]{Remark}

\numberwithin{equation}{section}
\newtheorem*{notation}{Notation}
\def \au {\rm}
\def \ti {\it}
\def \jou {\rm}
\def \bk {\it}
\def \no#1#2#3 {{\bf #1} (#3), #2.}
\def \eds#1#2#3 {#1, #2, #3.}
\title[Longterm Dynamics of the Extensible Suspension Bridge]
{Longterm Damped Dynamics \\ of  the Extensible Suspension Bridge}

\author[I. Bochicchio, C. Giorgi, E. Vuk]
{Ivana Bochicchio$^1$, Claudio Giorgi$^2$, Elena Vuk$^2$}
\address{
$^1$ Dipartimento di Matematica e Informatica, Universit\`a degli
studi di Salerno, Italy
\newline\indent
and INFN, Sez. di Napoli, Compl. Univ. di
Monte S. Angelo, Napoli, Italy} \email{ibochicchio@unisa.it}
\address{$^2$Dipartimento di Matematica, Universit\`a degli studi di Brescia, Italy}
\email{giorgi@ing.unibs.it} \email{vuk@ing.unibs.it}

\begin{document}

\begin{abstract}
This work is focused on the doubly nonlinear equation
\begin{equation}\nonumber
\ptt u+\partial_{xxxx}u + \big(p-\|\partial_x
u\|_{L^2(0,1)}^2\big)\partial_{xx}u+ \pt u +k^2u^+= f,
\end{equation}
whose solutions represent the bending motion of an extensible,
elastic bridge suspended by continuously distributed cables which
are flexible and elastic with stiffness $k^2$.\ When the ends are
pinned, longterm dynamics is scrutinized for arbitrary values of axial load
$p$ and stiffness $k^2$.\  For a general external source $f$ we prove the
existence of bounded absorbing sets.  When $f$ is time-independent,  the related semigroup of solutions is
shown to possess the global attractor of optimal regularity and its
characterization is given in terms of the steady
states of the problem.
\end{abstract}

\maketitle
%%%%%%%%%%%%%%%%%%%%%%%%%%%%%%%%%%%%%%%%%%%%%%%%%%%%%

%%%%%%%%%%%%%%%%%%%%%%%%%%%%%%%%%%%%%%%%%%%%%%%%%%%%%
\section{Introduction}

\subsection{The model equation}
In this paper, we scrutinize the longtime behavior of a nonlinear
evolution problem describing the damped oscillations of an
extensible elastic bridge of unitary natural length suspended by means of  flexible and elastic cables. The model equation ruling its dynamics can be derived from the standard modeling procedure, which relies on the basic assumptions of continuous distribution of the stays' stiffness along the girder and of the dominant truss behavior of the bridge (see, for instance, \cite{K-NYB}).

In the pioneer
papers by McKenna and coworkers (see
\cite{LMK,MKW,MKW1}),
the dynamics of  a suspension bridge is given by the well known damped equation
\begin{equation}
\label{BRIDGE} \ptt u+\partial_{xxxx}u + \pt u +k^2u^+= f,
\end{equation}
where $u=u(x,t):[0,1]\times\R\to\R$, accounts for the  downward deflection
of the bridge in the vertical plane, and $u^+$ stands for its positive part, namely,
\begin{equation}
u^+=
\begin{cases}
 u \qquad \text{if } u\geq0,\\
 0 \qquad \text{if } u<0.
\end{cases}
\end{equation}
Our model is derived here by taking into account the midplane
stretching of the road bed due to its elongation. As a consequence a
geometric nonlinearity appears into the bending equation.
This is achieved by
combining the pioneering ideas of Woinowsky-Krieger on the
extensible elastic beam  \cite{W} with equation
\eqref{BRIDGE}. Setting for simplicity all the
positive structural constants of the bridge equal to 1, we have
\begin{equation}
\label{BEAM} \ptt u+\partial_{xxxx}u + \big(p-\|\partial_x
u\|_{L^2(0,1)}^2\big)\partial_{xx}u+  \pt u +k^2u^+= f,
\end{equation}
where $f=f(x,t)$ is the (given) vertical dead load distribution.
The term $-k^2u^+$ models a restoring force
due to the cables, which is different from zero only when they are being stretched, and $\pt u$
accounts for an external resistant force linearly depending on the
velocity. The real constant $p$ represents the axial force acting
at the ends of the road bed of the bridge in the reference
configuration. Namely, $p$ is negative when the bridge is
stretched, positive when compressed.

As usual, $u$ and $\pt u$ are required to satisfy initial
conditions as follows
\begin{equation}
\label{BEAMIC}
\begin{cases}
u(x,0)=u_0(x), &  x\in[0,1],\\
\pt u(x,0)={u}_1(x), & x\in[0,1].
\end{cases}
\end{equation}
Concerning the boundary conditions, we consider here the case when
both ends of the bridge are pinned. Namely, for every $t\in\R$, we
assume
\begin{equation}
\label{BEAMBC}
u(0,t)=u(1,t)=\partial_{xx}u(0,t)=\partial_{xx}u(1,t)=0.
\end{equation}
This is the simpler choice. However, other types of boundary conditions with fixed ends are
consistent with the extensibility assumption as well; for instance,
when both ends are clamped, or when one end is clamped and the
other one is pinned. We address the reader to \cite{GPV} for a
more detailed discussion. Assuming \eqref{BEAMBC}, the domain of the differential operator $\partial_{xxxx}$ acting on $L^2(0,1)$ is
$$
\D(\partial_{xxxx})=\{w\in H^4(0,1) : w(0)=w(1)=\partial_{xx}w(0)=\partial_{xx}w(1)=0\}.
$$
This operator is strictly positive selfadjoint with compact
inverse, and its discrete spectrum is given by
$\lambda_n=n^4\pi^4$, $n\in\N$. Thus, ${\lambda_1=\pi^4}$ is the
smallest eigenvalue. Besides, the peculiar relation

$$(\partial_{xxxx})^{1/2}=-\partial_{xx}
$$
holds true, with Dirichlet boundary conditions and
$$\D(-\partial_{xx})=H^2(0,1)\cap H^1_0(0,1).$$
Hence, if pinned ends are considered, the initial-boundary value
problem \eqref{BEAM}--\eqref{BEAMBC} can be described by means of
a single operator $A=\partial_{xxxx}$, which enters the equation
at the powers 1 and 1/2. Namely,
$$
\partial_{tt} u+Au+ \pt u- \big(p-\|u\|^2_1\big)A^{1/2}u+k^2u^+=
f,$$
where $\| \cdot \|_1$ is the norm of $H^1_0(0,1)$.
This fact is particularly relevant in the
analysis of the {\it critical buckling load} $p_{\rm c}$, that is,
the magnitude of the compressive axial force $p>0$ at which
buckled stationary states appear.

As we shall show throughout the paper, this
model leads to exact results which are rather simple to  prove and, however,
are capable of capturing the main behavioral dynamic characteristics of the
bridge.

\subsection{Earlier contributions}
In recent years, an increasing attention was payed to the analysis of buckling, vibrations and post-buckling dynamics of nonlinear beam models, especially in connection with industrial applications \cite{LG,NP} and suspension bridges \cite{AGR1,AH}. As far as we know, most of the papers in the literature deal with approximations and numerical simulations, and only few works are able to derive exact solutions, at least under stationary conditions (see, for instance, \cite{BGV,BoV,CZGP,GV}). In the sequel, we give a brief sketch of earlier contributions on this subject.

In the fifties, Woinowsky-Krieger \cite{W}
proposed to modify the theory of the dynamic Euler-Bernoulli beam,
assuming a nonlinear dependence of the axial strain on the
deformation gradient. The resulting motion equation,
\begin{equation}
\label{BEAMW} \ptt u+\partial_{xxxx}u + \big(p-\|\partial_x
u\|_{L^2(0,1)}^2\big)\partial_{xx}u = 0,
\end{equation}
has been considered for hinged ends in the papers \cite{B1,D},
with particular reference to well-posedness results and to the
analysis of the complex structure of equilibria. Adding an
external viscous damping term $ \pt u$  to the original
conservative model, it becomes
\begin{equation}
\label{DAMPED} \ptt u+\partial_{xxxx}u+ \pt u+ \big(p-\|\partial_x
u\|_{L^2(0,1)}^2\big)\partial_{xx}u= 0.
\end{equation}
Stability properties of the unbuckled (trivial) and the buckled
stationary states of \eqref{DAMPED} have been established in \cite{B,D1} and, more formally, in \cite{RM}. In particular, if
$p<p_{\rm c}$, the exponential decay of solutions to the trivial
equilibrium state has been shown. The global dynamics of solutions
for a general $p$ has been first tackled in \cite{HAL} and improved in \cite{EM}, where the existence of a global attractor for \eqref{DAMPED}
subject to hinged ends was proved relying on the construction of a suitable
Lyapunov functional. In \cite{CZ} previous results are extended to a more general form of the nonlinear term by virtue of a suitable  decomposition of the semigroup first introduced in \cite{GPV}.

A different class of problems arises in the study of vibrations of a suspension bridge. The dynamic response of suspension bridges is usually analyzed by linearizing the equations of motion. When the effects of extensibility of the girder are neglected and the coupling with the main cable motion is disregarded, we obtain the well-known Lazer-McKenna equation \eqref{BRIDGE}.
Free and forced vibrations in models of this type, both with constant and non constant load, have been scrutinized in \cite{AH} and \cite{CJ}. The existence of strong solutions and global attractors for \eqref{BRIDGE} has been  recently obtained in \cite{ZMS}.

In certain cases Lazer-McKenna's model becomes inadequate and the effects of extensibility of the girder have to be taken into account. This can be done by introducing into the model equation \eqref{BRIDGE} a geometric nonlinear term like that appearing in \eqref{BEAMW}. Such a term is of some importance in the modeling of cable-stayed bridges (see, for instance, \cite{K-NYB,Vir}), where the elastic suspending cables are not vertical and produce a well-defined axial compression on the road bed.

Several studies have been devoted to the nonlinear vibrational analysis of mechanical models close to \eqref{BEAM}. Abdel-Ghaffar and Rubin \cite{AGR,AGR1}  presented a general theory and analysis of the nonlinear free coupled vertical-torsional vibrations of suspension
bridges. They developed approximate solutions by using the method of multiple scales via a perturbation technique. If torsional vibrations are ignored, their model reduces to \eqref{BEAM}.  Exact solutions to this problem, at least under stationary conditions, have been recently exhibited in \cite{GV}.

\subsection{Outline of the paper}

In the next Section 2, we formulate an abstract version of the
problem. We observe that its solutions are generated by a solution
operator $S(t)$, which turns out to be a strongly continuous semigroup in the autonomous case. The existence of an absorbing set for the solution
operator $S(t)$ is proved in Section 3 by virtue of a Gronwall-type Lemma. Section 4 is focused on the autonomous case and contains our main result. Namely, we establish the {\it existence of the regular global attractor} for a general $p$. In particular, we prove this by appealing to the existence of a Lyapunov functional and without requiring any assumption on the strength of the dissipation term.  A characterization of the global attractor is given in terms of the steady states of the system \eqref{BEAM}--\eqref{BEAMBC}.
First, we proceed with some preliminary estimates and prove the exponential stability of the system provided that the axial force $p$ is smaller than $p_c$. Finally, the smoothing property of the semigroup generated by the abstract problem is stated via a suitable decomposition first devised in \cite{GPV}.
%%%%%%%%%%%%%%%%%%%%%%%%%%%%%%%%%%%%%%%%%%%%%%%%%%%%%

%%%%%%%%%%%%%%%%%%%%%%%%%%%%%%%%%%%%%%%%%%%%%%%%%%%%%
\section{The Dynamical System}

\noindent In the sequel we recast problem~\eqref{BEAM}-\eqref{BEAMBC} into an abstract setting in order to establish more general results.

 Let $(H,\l\cdot,\cdot\r,\|\cdot\|)$ be a real Hilbert
space, and let $A:\D(A)\Subset H\to H$ be a strictly positive
selfadjoint operator with compact inverse. For $r\in\R$, we introduce the scale of
Hilbert spaces generated by the powers of $A$
$$
H^r=\D(A^{r/4}),\qquad \l u,v\r_r=\l A^{r/4}u,A^{r/4}v\r,\qquad
\|u\|_r=\|A^{r/4}u\|.
$$
When $r=0$, the index $r$ is omitted. The symbol
$\l\cdot,\cdot\r$ will also be used to denote the duality product
between $H^r$ and its dual space $H^{-r}$. In particular, we have
the compact embeddings $H^{r+1}\Subset H^r$, along with the
generalized Poincar\'e inequalities
\begin{equation}\label{POINCARE}
\lambda_1\|u\|_r^4\leq \|u\|_{r+1}^4,\qquad\forall u\in H^{r+1},
\end{equation}
where $\lambda_1>0$ is the first eigenvalue of $A$. Finally, we define the product Hilbert spaces
$$\H^r=H^{r+2}\times H^r.$$
For $p\in\R$, we consider the following abstract Cauchy problem on $\H$
in the unknown variable $u=u(t)$,
\begin{equation}
\label{ASTRATTO}
\begin{cases}
\partial_{tt} u+Au+ \pt u- \big(p-\|u\|^2_1\big)A^{1/2}u+k^2u^+=
f(t),\quad t>0,
\\
u(0)\,=\,u_0, \quad \pt u(0)\,=\,u_1 \ .
\end{cases}
\end{equation}
Problem~\eqref{BEAM}-\eqref{BEAMBC} is just a
particular case of the abstract system \eqref{ASTRATTO}, obtained
by setting $H=L^2(0,1)$ and $A=\partial_{xxxx}$ with the boundary
condition \eqref{BEAMBC}.

The following well-posedness result holds.

\begin{proposition}
\label{EU} Assume that $f\in L^1_{\rm loc}(0,T; H).$ Then, for all
initial data $z=(u_0,u_1)\in\H$, problem~\eqref{ASTRATTO} admits a
unique solution
$$(u(t),\pt u(t))\in\C(0,T; \H)\,,$$
which continuously depends on the initial data.
\end{proposition}

We omit the proof of this result, which is based on a standard Galerkin approximation procedure (see, for istance \cite{B,B1}), together with a slight generalization of the usual Gronwall lemma. In particular, the uniform-in-time estimates needed to obtain the global existence are exactly the same we use in proving the existence of an absorbing set.

In light of Proposition~\ref{EU},
we define the {\it solution operator}
$$S(t)\in\C(\H,\H),\qquad\forall t\geq 0,$$
as
$$z=(u_0,u_1)\mapsto S(t)z=(u(t),\pt u(t)).$$
Besides, for every $z\in\H$, the map $t\mapsto S(t)z$ belongs
to $\C(\R^+,\H)$.
Actually, it is a standard matter to verify the joint
continuity
$$(t,z)\mapsto S(t)z\in\C(\R^+\times\H,\H).$$

\begin{remark}
\label{semigroup}
In the autonomous case, namely when $f$ is time-independent, the
semigroup property
$$S(t+\tau)=S(t)S(\tau)$$
holds for all $t,\tau\geq 0$. Thus,
$S(t)$ is a strongly continuous semigroup of operators on $\H$ which continuously depends on the initial data: for any initial data $z\in\H$,
$S(t)z$ is the unique weak solution to \eqref{ASTRATTO}, with
related norm given by
\begin{equation*}
\mathcal{E}(z)\,=\,\|z\|_\mathcal{H}^2\,=\,\|u\|_2^2+\|v\|^2.
\end{equation*}
\end{remark}

For any $z=(u, v)\in\H$, we define the {\it energy} corresponding to $z$ as
\begin{equation} \label{eq:Energia}
E(z)=\mathcal{E}(z)+\frac12\big(\|u\|_1^2-p\big)^2+k^2\|u^+\|^2,
\end{equation}
and, abusing the notation, we denote $E(S(t)z)$ by $E(t)$ for each given initial data $z\in \H$.
Multiplying the first equation in \eqref{ASTRATTO} by $\pt u$, because of the relation
$$
k^2\l u^+, \pt u \r= \frac {k^2}{2} \frac{\,d}{dt}\,(\| u^+\|^2),
$$
we
obtain the {\it energy identity}
\begin{equation}
\label{E} \ddt E+2\|\partial_t u\|^2=2\l \pt u,f\r.
\end{equation}
In particular, for every $T>0$, there exists a positive increasing
function $\Q_T$ such that
\begin{equation}
\label{TFIN} E(t)\leq \Q_T(E(0)),\qquad\forall t\in [0,T].
\end{equation}
%%%%%%%%%%%%%%%%%%%%%%%%%%%%%%%%%%%%%%%%%%%%%%%%%%%%%

%%%%%%%%%%%%%%%%%%%%%%%%%%%%%%%%%%%%%%%%%%%%%%%%%%%%%
\section{The Absorbing Set}

\noindent It is well known that the absorbing set gives a first rough estimate of the dissipativity of the system. In addition, it is the
preliminary step to scrutinize its asymptotic dynamics (see, for instance, \cite{TEM}). Here, due to the joint presence of geometric and cable-response nonlinear terms in \eqref{ASTRATTO}, a direct proof of the existence of the absorbing set via {\it explicit} energy estimates is nontrivial.
Indeed, the double nonlinearity cannot be handled by means of standard arguments, as either in \cite{MKW} or in \cite{ZMS}.
Dealing with a given time-dependent external force $f$ fulfilling suitable translation compactness properties, a direct proof of the existence of an absorbing set is achieved here by means of a generalized Gronwall-type lemma devised in \cite{GGP}.

An absorbing
set for the solution operator $S(t)$ (referred to the initial time $t=0$) is a bounded set
$\BB_{\H} \subset \H$
with the following property: for every $R\geq 0$, there is an {\it
entering time} $t_R\geq 0$ such that
$$\bigcup_{t\geq t_R}S(t)z\subset \BB_{\H},
$$
whenever $\|z\|_\H\leq R$.
In fact, we are able to establish a more general result.

\begin{theorem}
\label{thmABS} Let $f\in L^\infty(\R^+,H)$, and let $\pt f$ be
a translation bounded function in $L^2_{\rm loc}(\R^+,H^{-2})$, that
is,
\begin{equation}
\label{tb} \sup_{t\geq 0}\int_t^{t+1}\|\pt f(\tau)\|^2_{-2}
\d\tau=M<\infty.
\end{equation}
Then, there exists $R_0>0$ with the following property:
in correspondence of every $R\geq 0$, there is $t_0=t_0(R)\geq 0$ such that
$$E(t)\leq R_0,\qquad \forall t\geq t_0,
$$
whenever $E(0)\leq R$. Both $R_0$ and $t_0$ can be explicitly
computed.
\end{theorem}

We are able to establish Theorem \ref{thmABS}, leaning on the following Lemma.

\begin{lemma}[see Lemma 2.5 in \cite{GGP}]
\label{superl}
Let $\Lambda:\R^+\to\R^+$ be an absolutely continuous function
satisfying, for some $M\geq 0$, $\eps>0$,
the differential inequality
$$
\ddt\Lambda(t)+\eps \Lambda(t)\leq \varphi(t),
$$
where $\varphi:\R^+\to\R^+$ is any locally summable function such that
$$\sup_{t\geq 0}\int_t^{t+1}\varphi(\tau)\d \tau\leq M.$$
Then, there exist $R_1>0$ and $\gamma>0$ such that,
for every $R\geq 0$, it follows that
$$\Lambda(t)\leq R_1,\qquad\forall t\geq R^{1/\gamma}(1+\gamma M)^{-1},$$
whenever $\Lambda(0)\leq R$.
Both $R_1$ and $\gamma$ can be explicitly computed
in terms of  $M$ and $\eps$.
\end{lemma}

\begin{proof}[Proof of Theorem \ref{thmABS}]
Here and in the sequel, we will tacitly use several times the
Young and the H\"older inequalities, besides the usual Sobolev
embeddings. The generic positive constant $C$ appearing in this
proof may depend on $p$ and $\|f\|_{L^\infty(\R^+,H)}$.

On account of \eqref{E}, by means of the functional
\begin{equation}\label{elle}
\L(z)=E(z)-2\l u,f\r,
\end{equation}
we introduce the function
$$
\L(t)=E(t)-2\l u(t),f(t)\r,
$$
which satisfies the differential equality
\begin{equation}\label{derivata}
\ddt \L+2\|\partial_t u\|^2=-2\l u,\pt f \r.
\end{equation}
Because of the control
\begin{equation}
\label{derivataf}
2|\l u,\pt f\r| \leq \frac12 \|u\|_2^2 + 2\| \pt f\|^2_{-2},
\end{equation}
we obtain the differential inequality
\begin{equation*}
\ddt \L+2\|\partial_t u\|^2 \leq \frac 12 E + 2\| \pt f\|^2_{-2}.
\end{equation*}
Next, we consider the auxiliary functional $\Upsilon(z)=\l u,v\r$ and,
regarding  $\Upsilon(S(t)z)$ as $\Upsilon(t)$, we have
$$
\ddt\Upsilon+\Upsilon + \|u\|_2^2+\big(\|u\|_1^2-p\big)^2
+p\big(\|u\|_1^2-p\big) + k^2\|u^+\|^2 - \l
u,f\r=\|\pt u\|^2.
$$
Noting that
$$
\frac12\big(\|u\|_1^2-p\big)^2
+p\big(\|u\|_1^2-p\big)=\frac12\|u\|_1^4-\frac12 p^2,
$$
we are led to
\begin{equation*}
\ddt\Upsilon + \Upsilon + \frac12\|u\|_1^4 + \|u\|_2^2 + k^2\|u^+\|^2+\frac12\big(\|u\|_1^2-p\big)^2 - \l u,f\r
= \|\pt u\|^2+\frac12 p^2.
\end{equation*}
Precisely, we end up with
\begin{equation}
\label{UNODUE}
\ddt\Upsilon + \Upsilon + \frac12 E\leq \frac32 \|\partial_t
u\|^2+\frac{1}{2 \lambda_1}\|f\|^2+\frac{1}{2}p^2.
\end{equation}
Finally, we set
$$\Lambda(z)=\L(z)+\Upsilon(z)+C,$$
where $ C= \frac{2}{\lambda_1}\|f\|^2+\frac{1}{2\lambda_1} + \frac{|p|}{2\sqrt {\lambda_1}}$.
We first observe that $\Lambda(z)$ satisfies
\begin{equation}
\label{CTRL} \frac12 \E(z)\leq \frac12 E(z)\leq\Lambda(z)\leq 2 E(z)+c.
\end{equation}
In order to estimate $\Lambda$ from below, a straightforward calculation leads to
\begin{eqnarray*}
\Lambda(z)&\geq &E(z) - 2 |\l u,f\r | - | \Upsilon(z)| + C
\\
&\geq& E(z) - \frac12\|u\|_2^2 -  \frac 12\|v\|^2- \frac14\big(\|u\|_1^2-p\big)^2 -2 \|f\|_{-2}^2 - \frac{1}{2\lambda_1} - \frac{|p|}{2\sqrt {\lambda_1}} + C
\\
&=&\frac 12 \|u\|_2^2 + \frac 12 \|v\|^2 + \frac14\big(\|u\|_1^2-p\big)^2 + k^2\|u^+\|^2 -  2 \|f\|_{-2}^2 - \frac{1}{2\lambda_1} - \frac{|p|}{2\sqrt {\lambda_1}} + C
\\
&\geq& \frac 12 E(z) - \frac{2}{\lambda_1}\|f\|^2-\frac{1}{2\lambda_1} - \frac{|p|}{2\sqrt {\lambda_1}} +C\geq  \frac 12 E(z),
\end{eqnarray*}
where we take advantage of
\begin{align*}
 | \Upsilon(z)|&\leq \|u\|\|v\| \leq \frac{1}{\root 4 \of{ \lambda_1}} \|u\|_1\|v\| \leq \frac 12\|v\|^2 + \frac{1}{2\sqrt {\lambda_1}} \|u\|_1^2  \\
 &\leq \frac 12\|v\|^2  +  \frac{1}{2\sqrt {\lambda_1}}\big(\|u\|_1^2-p\big) + \frac{|p|}{2\sqrt {\lambda_1}} \\
 &\leq \frac 12\|v\|^2  +  \frac14 \big(\|u\|_1^2-p\big)^2 + \frac{|p|}{2\sqrt {\lambda_1}} +  \frac{1}{2\lambda_1}.
\end{align*}
The upper bound for $\Lambda$ can be easly achieved as follows
\begin{align*}
\Lambda(z)&\leq E(z) + \|u\|_2^2 + \|v\|^2+ \frac12 \|u\|_1^4 + \|f\|_{-2}^2 + \frac{1}{32\lambda_1} + C \\
&\leq 2E(z) + \frac{1}{\lambda_1}\|f\|^2 + \frac{1}{32\lambda_1} + C \leq 2E(z) + c,
\end{align*}
by virtue of
\begin{equation}\label{Ups2}
 | \Upsilon(z)|\leq \|u\|\|v\| \leq \frac{1}{\root 4 \of{ \lambda_1}} \|u\|_1\|v\| \leq \|v\|^2 + \frac{1}{4\sqrt {\lambda_1}} \|u\|_1^2 \leq  \|v\|^2  +  \frac12 \|u\|_1^4 + \frac{1}{32\lambda_1}.
\end{equation}
Going back to differential equation and making use of \eqref{derivata} and \eqref{UNODUE}, the function $\Lambda(t)= \Lambda (S(t)z)$  satisfies the identity
$$
\ddt\Lambda + \frac {\Lambda}{2} +  \frac {\Upsilon}{2} + \frac12\|u\|_2^2 + \frac12\|\partial_tu\|^2 +
\frac14\|u\|_1^4 + \frac12k^2\|u^+\|^2 =
-2\l u,\partial _tf\r +\frac{p^2}{2},
$$
and, as a consequence, we obtain the estimate
\begin{equation*}
\ddt\Lambda + \frac {\Lambda}{2} +  \frac 12 \big(\Upsilon + \|u\|_2^2 + \|\partial_tu\|^2 +
\frac12\|u\|_1^4  +4 \l u,\partial _tf\r \big) \leq \frac{p^2}{2}.
\end{equation*}
Now, using \eqref{Ups2} and \eqref{derivataf}, we have
$$
\ddt\Lambda + \frac {\Lambda}{2} \leq 2\| \partial _tf\|_{-2}^2 + c,
$$
where $c= \frac {1}{16 \lambda_1} + \frac{p^2}{2}$.
Thus, by virtue of \eqref{tb} and \eqref{CTRL}, Lemma~\ref{superl} yields
$$
E(t) \leq 2\Lambda(t) \leq 2R_1(M,c).
$$
\end{proof}

\begin{remark}
If the set of stationary solutions to \eqref{ASTRATTO}  shrinks to a single element, the subsequent asymptotic behavior of the system becomes quite simple. Indeed, this occurs when $p< p_c=\sqrt{\lambda_1}$. If this the case, the only trivial solution exists and is exponentially stable,  as it will be shown in Section 4. The more complex and then attractive situation occurs when the set of steady solutions contains a large (possibly infinite) amount  of elements. To this end, we recall here that the set of the bridge stationary-solutions (equilibria) has a very rich structure, even when $f=0$ (see \cite{GV}).
\end{remark}

%%%%%%%%%%%%%%%%%%%%%%%%%%%%%%%%%%%%%%%%%%%%%%%%%%%%%

%%%%%%%%%%%%%%%%%%%%%%%%%%%%%%%%%%%%%%%%%%%%%%%%%%%%%
\section{The Global Attractor}

In the remaining of the paper, we simplify the problem by assuming that the external force $f$ is time-independent. In which case, the operator $S(t)$ is a strongly continuous semigroup on $\H$ (see Remark\,\ref{semigroup}). Having been proved in Sect.\,3 the existence of the absorbing set $\BB$, we could then establish here the existence of a global attractor by showing that the semigroup $S(t)$ admits a bounded absorbing set in a more regular space and that it is uniformly compact for large values of $t$ (see, for instance, \cite[Theor.\,1.1]{TEM}). In order to obtain asymptotic compactness, the $\alpha$-contraction method should be employed  (see \cite{HAL} for more details).
If applied to \eqref{ASTRATTO}, however, such a strategy would need a lot of calculations and, what is more, would provide some regularity of the attractor only if the dissipation is large enough (see \cite{EM}, for instance).

Noting that in the autonomous case problem \eqref{ASTRATTO} becomes a {\it gradient system}, there is a way to overcome these difficulties by using an alternative approach which appeals to the existence of a Lyapunov functional in order to prove the existence of a global attractor. This technique has been successfully adopted in some recent papers  concerning some related problems, just as the longterm analysis of the transversal motion of extensible viscoelastic \cite{GPV} and thermoelastic \cite{GNPP} beams.

We recall that the global attractor $\A$ is the unique compact
subset of $\H$ which is at the same time
\begin{itemize}
\item[(i)] attracting:
$$\lim_{t\to\infty}\boldsymbol{\delta}(S(t)\BB,\A)\to 0,$$
for every bounded set $\BB\subset\H$, where $\boldsymbol{\delta}$
denotes the usual Hausdorff semidistance in $\H$; \item[(ii)]
fully invariant:
$$S(t)\A=\A,\qquad\forall t\geq 0.$$
\end{itemize}
We address the reader to the books \cite{CV,HAL,TEM} for a
detailed presentation of the theory of attractors.

\begin{theorem}
\label{MAIN} The semigroup $S(t)$ acting on $\H$ possesses a connected
global attractor $\A$ bounded in $\H^2$. Moreover, $\A$ coincides
with the unstable manifold of the set ${\mathcal S}$ of the
stationary points of $S(t)$, namely,
$$\A=
\Big\{z(0): z \text{ is a complete (bounded) trajectory of } S(t)
: \lim_{t\to \infty}\|z(-t)-{\mathcal
S}\|_{\H}=0\Big\}.
$$
\end{theorem}

\begin{remark}
Due to the regularity and the invariance of $\A$, we observe that
$S(t)z$ is a strong solution to \eqref{ASTRATTO} whenever
$z\in\A$.
\end{remark}

The set ${\mathcal S}$ of the bridge equilibria under a vanishing lateral load consists of all the pairs $(u,0)\in \H$ such that the function $u$ is a weak solution to the equation
$$Au - \big(p-\|u\|^2_1\big)A^{1/2}u+k^2u^+=0.$$
In particular,
$u$ solves the following boundary value problem on
the interval $[0,1]$
\begin{equation}
\label{STATICBEAM}
\begin{cases}
\partial_{xxxx}u+\big(b\,\pi^2-\|\partial_{x}u\|^2_{L^2(0,1)}\big)\partial_{xx}u+\kappa^2\pi^4u^+= 0,\\
\noalign{\vskip1mm} u(0)=u(1)=\partial_{xx}u(0)=\partial_{xx}u(1)=0,
\end{cases}
\end{equation}
where we let $k=\kappa\pi^2$, $\kappa\in\R$, and $p=b\,\pi^2$,
$b\in\R$. It is then apparent that ${\mathcal S}$ is bounded in
$H^2(0,1)\cap H^1_0(0,1)$ for every $b,\kappa\in\R$.

When $\kappa=0$, a general result has been established in
\cite{CZGP} for a class of non-vanishing sources. In \cite{BGV, BoV},
the same strategy with minor modifications has been applied  to
problems close to \eqref{STATICBEAM}, where the term $u^+$ is
replaced by $u$ (unyielding ties).

 The set of buckled solutions to problem \eqref{STATICBEAM} is built up and scrutinized in \cite{GV}.
In order to have a finite number of solutions, we need
all the bifurcation values to be
distinct. This occurrence trivially holds when $\kappa=0$, because
of the spectral properties of the operator $\partial_{xxxx}$. On
the contrary, for general values of $\kappa$, all critical values
``moves" when $\kappa$ increases, as well as in \cite{BoV}. Hence, it may happen that two different
bifurcation values overlap for special values of $\kappa$, in which case they are referred as {\it resonant values}.

\begin{figure}[ht]
\begin{center}
\includegraphics[width=15cm, height=5cm]{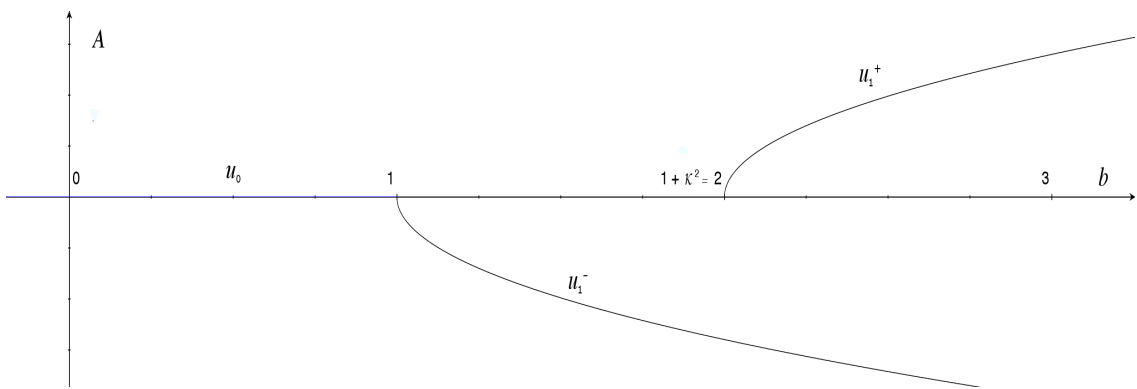}
\caption{The bifurcation picture for $ \kappa=1$.}
\label{Fig_1}
\end{center}
 \end{figure}

\noindent
Assuming that $\kappa=1$, for istance,  Fig.~1 shows the bifurcation picture of solutions in dependence on the applied axial load $p=b\pi^2$. In particular,
$u_0=0$ and
$$
u_1^\pm(x)=A_1^\pm\sin(\pi x),
\qquad
A_1^-=-\sqrt{2(b-1)\,},
\ A_1^+=\sqrt{2(b-2)\,}.
$$

%%%%%%%%%%%%%%%%%%%%%%%%%%%%%%%
%%%%%

\subsection{The Lyapunov Functional and preliminary estimates}

We begin to prove the existence of a Lyapunov
functional for $S(t)$, that is, a function $\L\in C(\H,\R)$ satisfying the
following conditions:
\begin{itemize}
\item[(i)] $\L(z)\to\infty$ if and only if $\|z\|_{\H}\to\infty$;
\item[(ii)] $\L(S(t)z)$ is nonincreasing for any $z\in\H$;
\item[(iii)] $\L(S(t)z)=\L(z)$ for all $t>0$ implies that
$z\in{\mathcal S}$.
\end{itemize}

\begin{proposition}
\label{LYAP} If $f$ is time-independent, the functional $\L$ defined in \eqref{elle} is a Lyapunov functional for $S(t)$.
\end{proposition}

\begin{proof}
Assertion (i) holds by the continuity of $\L$ and by means of the estimates
$$\frac12 E(z) - c \leq\L(z)\leq \frac 32 E(z)+c.$$
Using \eqref{derivata}, we obtain quite directly
\begin{equation} \label{LiapDis}
\frac{d}{dt}\L(S(t)z) =-2\|\pt u(t)\|^2 \leq 0,
\end{equation}
which proves the decreasing monotonicity of $\L$ along the trajectories
departing from $z$. Finally, if $\L(S(t)z)$ is constant in time, we have that $\pt
u=0$ for all $t$, which implies that $u(t)$ is constant. Hence,
$z=S(t)z=(u_0,0)$ for all $t$, that is, $z\in{\mathcal S}$.
\end{proof}

The existence of a Lyapunov functional ensures that $E(t)$ is
bounded. In particular, bounded sets have bounded orbits.

\begin{notation}
Till the end of the paper, $Q: \R^+_0\to\R^+$ will denote a {\it generic} increasing monotone function depending explicity only on $R$ and implicity on the structural constants of the problem. The actual expression of $Q$ may change, even within the same
line of a given equation.
\end{notation}

\begin{lemma}
\label{energia-limitata} Given $f\in H$, for all $t>0$ and initial data $z\in \H$ with $\|z\|_{\H}\leq R$,
\begin{equation}
\label{energy} \E(t)\leq Q(R).
\end{equation}
\end{lemma}
\begin{proof}
Inequality \eqref{LiapDis} ensures that
\begin{equation} \nonumber
\begin{split}
\L(t)=\L(S(t)z) \leq \L(z)\leq Q(R) \ , \quad \forall t \geq 0.
\end{split}
\end{equation}
Moreover, taking into account that
$$\left\| u(t)\right\| ^{2} \leq \frac{1}{\lambda_1}\left\| u(t)\right\|_{2} ^{2} \leq \frac{1}{ \lambda_1}\E(t),$$
we obtain the estimate
$$\L(t) \geq \E(t) - 2\left\langle f\,,\,u(t)\right\rangle \geq \E(t) - 2 \left\| f\right\|_{-2} ^{2} - \frac 12 \left\| u(t)\right\|_{2} ^{2} \geq
\frac 12 \E(t) - \frac{2}{\lambda_1} \left\|
f\right\| ^{2}.$$
Finally, we have
$$\E(t)\leq 2\L(t) +
 \frac{4}{\lambda_1} \left\| f\right\| ^{2} \leq 2Q(R) + \frac{4}{\lambda_1} \left\| f\right\| ^{2} = Q(R).$$
\end{proof}

\begin{lemma} \label{FunzionaleF}
Let $p <\, \sqrt{\lambda_1}$ and
$
\mathcal F_p(u)\,=Au\,-\,p A^{\frac{1}{2}}u\,
$.
Then
$$\left\langle \mathcal F_p(u)\,,\,u\right\rangle \geq C(p) \left\| u\right\| _{2}^{2}, $$
where
\begin{equation}
\label{Cp}
C(p)=
\begin{cases}
1, \qquad\qquad\quad p\leq 0 \\
\left(1- \frac {p}{\sqrt{\lambda_1}}\right), \quad 0<p<\sqrt{\lambda_1}.
\end{cases}
\end{equation}
\end{lemma}
\begin{proof}
Because of the identity
$$\left\langle \mathcal F_p(u)\,,\,u\right\rangle = \left\| u\right\|
_{2}^{2}-\,p \left\| u\right\| _{1}^{2},$$
the thesis is trivial when $p \leq 0$.
On the other hand, when $0<p<\sqrt{\lambda_1}$ we have
$$\left\langle \mathcal F_p(u)\,,\,u\right\rangle = \left\| u\right\|
_{2}^{2}-p \, \left\| u\right\| _{1}^{2}\,\geq \left( 1-\frac{p
}{\sqrt{\lambda_1}}\right) \left\|
u\right\| _{2}^{2} .$$
\end{proof}
We are now in a position to prove the following
\begin{theorem} \label{exp-stab}
When $f\,=0$, the solutions to \eqref{BEAM}--\eqref{BEAMBC} decay
exponentially, i.e.
$$
\E(t)\,\leq\,c_0\,\E(0)\,e^{- c t}
$$
with $c_0$ and $c$ suitable positive constants, provided that
$p\,<\, \sqrt{\lambda_1}$.
\end{theorem}

\begin{proof}
Let ${\Phi}$ be the functional
$$\Phi(z)= \E(z) + \varepsilon \Upsilon(z) - \frac 12 p^2,
$$
where the constant
\begin{equation} \label{asterisco}
\varepsilon = {\rm min} \{\lambda _{1}C(p), 1\}
\end{equation}
is positive provided that $p <\sqrt{\lambda_1}$. In view of applying Lemma \ref{FunzionaleF}, we remark that
$$
\Phi =\left\langle \mathcal F_p(u)\,,\,u\right\rangle+\left\|
\partial_t u\right\| ^{2}+\frac{1}{2}\left\| u\right\| _{1}^{4}+\varepsilon
\left\langle u,\partial_t u\right\rangle+k^2\left\| u^+\right\|
^{2}.
$$
The first step is to prove the equivalence between $\mathcal{E}$
and $\Phi$, that is
\begin{equation} \label{Phi_bounds}
\frac \varepsilon{2\,\lambda_1}\,\mathcal{E}\leq \Phi\leq Q(\|z\|_{\H})\,\,\mathcal{E}\,.
\end{equation}
By virtue of \eqref{POINCARE}, \eqref{asterisco} and Lemma \ref{FunzionaleF} the lower bound is provided by
$$
\Phi\geq  \left( C(p) -\frac{\varepsilon }{2\,\lambda _{1}
}\right) \left\| u\right\| _{2}^{2}+\left( 1-\frac{\varepsilon }{2\,}\right) \left\|
\partial_t u\right\| ^{2} \geq \frac \varepsilon{2\,\lambda_1}\,\E\,.
$$

On the other hand, by applying Young inequality and using
(\ref{POINCARE}), we can write the following chain of inequalities which gives the upper bound of $\Phi$.
\begin{eqnarray*}
\Phi &\leq&
 \left( C(p)+\frac{k^2}{ \lambda
_{1}}+\frac{1}{2 \,\lambda _{1}}\right) \left\| u\right\|
_{2}^{2}+\left( 1+\frac{\varepsilon ^{2}}{2}\right) \left\|
\partial_t u\right\| ^{2}- p \left\| u\right\|
_{1}^{2}+\frac{1}{2}\left\| u\right\| _{1}^{4}\leq
\\
&\leq &\left( 1+ C(p)+\frac{k^2}{ \lambda _{1}}+\frac{1}{2 \,\lambda
_{1}}+\frac{\varepsilon ^{2}}{2}\right) \E + \left\| u\right\|
_{1}^{2}\left(\frac{1}{2}\left\| u\right\| _{1}^{2}-p\right) .
\end{eqnarray*}
In particular, from \eqref{energy} and \eqref{Cp}
 we find
$$
\Phi\leq \left( 2+\frac{k^2}{ \lambda _{1}}+\frac{1}{2 \,\lambda
_{1}}+\frac{\varepsilon ^{2}}{2}+\frac{Q(\|z\|_{\H})}{ \sqrt{\,\lambda
_{1}}} \,\right) \E  = Q(\|z\|_{\H})\,\,\mathcal{E}.
$$
The last step is to prove the exponential decay of $\Phi$.
To this aim, we obtain the identity
\begin{equation} \nonumber
\frac{d}{dt}\, \Phi+\varepsilon \Phi+2\left( 1
-\varepsilon \right) \left\| \partial_t u\right\|
^{2}+\frac{\varepsilon }{2}\left\| u\right\| _{1}^{4}+\varepsilon
\left( 1 -\varepsilon \right) \left\langle \partial_t
u\,,\,u\right\rangle =0,
\end{equation}
where $\eps$ is given by \eqref{asterisco}.
%%%
Exploiting the Young inequality and  \eqref{Phi_bounds}, we have
\begin{equation} \nonumber
\frac{d}{dt}\, \Phi+\varepsilon \Phi+ \left(
1 -\varepsilon \right) \left\| \partial_t u\right\| ^{2}\leq
\frac{\varepsilon^2\left( 1 -\varepsilon \right)}{4\lambda_{1}
}\,\left\| u\right\| ^{2}_{2} \ \leq \frac{\varepsilon\left( 1 -\varepsilon \right)}{2}\,\Phi,
\end{equation}
from which it follows
\begin{equation} \nonumber
\frac{d}{dt}\, \Phi+ \frac{\varepsilon\left( 1 +\varepsilon \right)}{2} \Phi\leq 0.
\end{equation}
Letting $c = {\varepsilon\left( 1 +\varepsilon \right)}/{2}$, by virtue of Lemma \ref{superl} (with $M=0$) and \eqref{Phi_bounds} we have
$$\frac \varepsilon{2\,\lambda_1}\,\E(t)\leq\Phi\left( t\right) \leq \Phi\left( 0\right)
\,e^{- \,c\,t}\leq Q(\|z\|_{\H})\,\E\left( 0\right)\,e^{- \,c\,t}.$$
The thesis follows by putting $c_{0}=2\,\lambda_1Q(\|z\|_{\H})/\varepsilon$.
\end{proof}

\par
The existence of a Lyapunov functional, along with the fact that
${\mathcal S}$ is a bounded set, allow us prove the existence of
the attractor by showing a
suitable (exponential) asymptotic compactness property of the
semigroup, which will be obtained exploiting a particular
decomposition of $S(t)$ devised in \cite{GPV} and following a general result (see \cite{CP}, Lemma 4.3), tailored to our particular case.

%%%%%%%%%%%%%%%%%%%%%%%%%%%%%%%%%%%%%%%%%%%%%%%%%%%%%%%%%%%%%%%%%%%%%%

\subsection{The Semigroup Decomposition}

\noindent By the interpolation inequality $\|u\|_1^2\leq \|u\|\|u\|_2$ and (\ref{POINCARE})
it is clear that
\begin{equation}
\label{gamma} \frac12\|u\|_2^2\leq
\|u\|_2^2-p\|u\|_1^2+\alpha\|u\|^2\leq m\|u\|_2^2,
\end{equation}
provided that $\alpha>0$ is large enough and for some $m=m(p,\alpha)\geq 1$.

Again, $R>0$ is fixed and $\|z\|_{\H}\leq R$. Choosing $\alpha>0$
such that \eqref{gamma} holds, according to the scheme first proposed in \cite{GPV}, we decompose the solution $S(t)z$
into the sum
$$S(t)z=L(t)z+K(t)z,$$
where
$$L(t)z=(v(t),\pt v(t))\qquad\text{and}\qquad
K(t)z=(w(t),\pt w(t))$$
solve the systems
\begin{equation}
\label{DECAY}
\begin{cases}
\ptt v+Av+\pt v-
(p-\|u\|^2_1)A^{1/2}v+\alpha v= 0,\\
\noalign{\vskip1.5mm} \left(v(0),\pt v(0)\right)=z,
\end{cases}
\end{equation}
and
\begin{equation}
\label{CPT}
\begin{cases}
\ptt w+Aw+\pt w-
(p-\|u\|^2_1)A^{1/2}w-\alpha v+k^2u^+= f,\\
\noalign{\vskip1.5mm} (w(0),\pt w(0))=0.
\end{cases}
\end{equation}

\noindent The next three lemmas show the asymptotic smoothing property of
$S(t)$, for initial data bounded by $R$. We begin to prove the
exponential decay of $L(t)z$. Then, we prove the asymptotic
smoothing property of ${K}(t)$.

\begin{lemma}
\label{lemmaDECAY} There is $\omega=\omega(R)>0$ such that
$$\|L(t)z\|_{\H}\leq Ce^{-\omega t}.$$
\end{lemma}

\begin{proof}
After denoting
$$\E_0(t)=\E_0(L(t)z)=\|L(t)z\|_{\H}^2= \|v(t)\|^2_2+\|\pt v(t)\|^2,$$
we set $\Phi_0(t)=\Phi_0(L(t)z,u(t))$, where $u(t)$ is the first component of $S(t)z$ and
$$
\Phi_0(L(t)z,u(t))=\E_0(L(t)z)- p \|v(t)\|_1^2+ (\alpha + \frac 12) \|v(t)\|^2
+\|u(t)\|^2_1\|v(t)\|^2_1+\l\pt v(t),v(t)\r.
$$
In light of Lemma \ref{energia-limitata} and inequalities (\ref{gamma}),
we have the bounds
\begin{equation}
\label{stima-3} \frac12\E_0\leq \Phi_0\leq Q(R)\E_0.
\end{equation}
Now, we compute the time-derivative of $\Phi_0$ along the solutions to system (\ref{DECAY}) and we obtain
$$
\frac{d}{dt}\Phi_0+\Phi_0=2\l\pt
u,A^{1/2}u\r\|v\|^2_1\leq Q(R)\|\pt u\|\Phi_0.
$$
The exponential decay of $\Phi_0$ is entailed by exploiting the following Lemma \ref{lemmaINT} and then applying Lemma 6.2 of \cite{GPV}. From \eqref{stima-3} the desired decay of $\E_0$ follows.
\end{proof}

\begin{lemma}
\label{lemmaINT}
For any $\eps > 0$
$$
\int\limits_{\tau }^{t}\left\|
\partial_t u\left( s\right) \right\| ds\leq \eps (t-\tau
)+\frac \eps 4+ \frac{Q(R)}{\eps }\,,
$$
for every $t\geq\tau\geq0.$
\end{lemma}

\begin{proof}
After integrating \eqref{LiapDis} over $(\tau, t)$ and taking \eqref{energy} into account, we obtain
$$
\frac 12 \E(S(t)z) -2 \frac{\|f\|^2}{\lambda_1} \leq \L(S(t)z) + 2 \int\limits_{\tau }^{t}\left\|
\partial_t u\left( s\right) \right\|^2 ds = \L(S(\tau)z) \leq \L(z).
$$
It follows
$$ \int\limits_{\tau }^{t}\left\|
\partial_t u\left( s\right) \right\|^2 ds \leq Q(R),$$
which, tanks to the H\"{o}lder inequality, yields
$$
\int\limits_{\tau }^{t}\left\| \partial_t u\left( s\right)
\right\| ds\leq
\eps\sqrt{t-\tau } + \frac{Q(R)}{\eps} \leq
\eps (t-\tau
)+\frac \eps 4+ \frac{Q(R)}{\eps },
$$
for any  $\eps>0$.
\end{proof}

The next result provides the boundedness of $K(t)z$ in a more
regular space.

\begin{lemma}{\rm{(see \cite{GPV}, Lemma 6.3)}}
\label{lemmaCPT} The estimate
$$\|K(t)z\|_{\H^2}\leq Q(R)$$
holds for every $t\geq 0$.
\end{lemma}

\begin{proof}
As well as in \cite{GPV}, we use here the interpolation inequality
$$
\|w\|_3^2\leq\|w\|_2\|w\|_4.
$$
Jointly with $\|w\|_2\leq Q(R)$ (which follows by comparison
from (\ref{energy}) and Lemma~\ref{lemmaDECAY}), this entails
\begin{equation} \label{interp}
p\|w\|_3^2\leq \frac12\E_1+Q(R),
\end{equation}
where
$$\E_1(t)=\E_1(K(t)z)=\|K(t)z\|_{\H^2}^2=\|w(t)\|^2_4+\|\pt w(t)\|_2^2.$$
Letting
$$\Phi_1=\E_1
+(\|u\|^2_1-p)\|w\|^2_3+\l\pt w,Aw\r-2\l f,A
w\r + 2 k^2 \l u^+,A w\r ,$$
we have the bounds
\begin{equation}
\label{Phi1_bounds}
\frac13\E_1-Q(R)\leq \Phi_1\leq Q(R)\E_1+Q(R).
\end{equation}
Taking the time-derivative of $\Phi_1$, we find
\begin{align*}
&\frac{d}{dt}\Phi_1+\Phi_1=
2\l\pt u,A^{1/2}u\r\|w\|^2_3+ 2\alpha \l A^{1/2}v,A^{1/2}\pt w\r +\\
&\quad+\Big[\alpha \l A^{1/2}v,A^{1/2}w\r-\l f,Aw\r\Big] - k^{2} \left\langle Aw,u^{+}\right\rangle
+2k^{2}\left\langle Aw,\partial _{t}u^{+}\right\rangle \ .
\end{align*}
Using (\ref{energy}) and (\ref{interp}), we control the rhs by
$$
\frac18\E_1+Q(R)\sqrt{\E_1}\,+Q(R)\leq\frac14\E_1 +Q(R) \leq
\frac{3}{4}\Phi_1 +Q(R) ,
$$
and we obtain
$$
\frac{d}{dt}\Phi_1+\frac{1}{4}\Phi_1 \leq Q(R).
$$
Since $\Phi_1(0)=0$, the standard Gronwall lemma yields the boundedness of $\Phi_1$. Then, by virtue of \eqref{Phi1_bounds}, we obtain the desired estimate for $\E_1$.
\end{proof}

By collecting previous results, Lemma 4.3 in \cite{CP} can
be applied to obtain the existence of the attractor $\A$ and its regularity.
Within our hypotheses and by virtue of the decomposition \eqref{DECAY}-\eqref{CPT}, it is also possible to prove the existence of regular exponential attractors for $S(t)$ with finite fractal dimension in $\H$. This can be done by a procedure very close to that followed in \cite{GPV}.
Since the global attractor is the {\it minimal} closed attracting
set, we can conclude that the fractal dimension of $\A$ in $\H$ is
finite as well.

\medskip
%%%%%%%%%%%%%%%%%%%%%%%%%%%%%%%%%%%%%%%%%%%%%%%%%%%%%
\noindent
{\ti Acknowledgments}

The authors are indebted to the anonymous referees for their valuable remarks and comments.

%%%%%%%%%%%%%%%%%%%%%%%%%%%%%%%%%%%
%%%%%%%%%%%%%%%%%%

%%%%%%%%%%%%%%%%%%%%%%%%%%%%%%%%%%%%%%%%%%%%%%

\end{document}